\documentclass[11pt]{article}
\pdfoutput=1

\usepackage[utf8]{inputenc} 
\usepackage{url}            
\usepackage{booktabs}       
\usepackage{amsfonts}       
\usepackage{authblk}

\usepackage{subcaption}

\usepackage[a4paper]{geometry}

\usepackage{amsmath, amssymb, amsthm}
\usepackage{tikz}
\newcommand{\reals}{\mathbb{R}}

\DeclareMathOperator*{\argmin}{arg\,min}

\global\long\def\data{y}
\global\long\def\logit{f}
\global\long\def\cova{\xi}
\global\long\def\param{x}

\usepackage{url}

  \theoremstyle{plain}
  \newtheorem*{prop*}{\protect\propositionname}
\theoremstyle{plain}
\newtheorem{prop}{\protect\propositionname}

\providecommand{\propositionname}{Proposition}
\providecommand{\theoremname}{Theorem}

\begin{document}

\title{Piecewise Deterministic Markov Processes for Scalable Monte Carlo on Restricted Domains}

\author[1]{Joris~Bierkens \thanks{\texttt{joris.bierkens@tudelft.nl}; Corresponding Author}}
\author[2]{Alexandre~Bouchard-C\^{o}t\'{e}}
\author[3]{Arnaud~Doucet}
\author[4]{Andrew~B.~Duncan}
\author[5]{Paul~Fearnhead}
\author[3]{Thibaut~Lienart}
\author[6]{Gareth~Roberts}
\author[6]{Sebastian~J.~Vollmer}

\affil[1]{Delft Institute of Applied Mathematics, TU Delft, Netherlands}
\affil[2]{Department of Statistics, University of British Columbia, Canada}
\affil[3]{Department of Statistics, University of Oxford, UK}
\affil[4]{School of Mathematical and Physical Sciences, University of Sussex, UK}
\affil[5]{Department of Mathematics and Statistics, Lancaster University, UK}
\affil[6]{Department of Statistics, University of Warwick, UK}



\maketitle

\begin{abstract}
 
Piecewise Deterministic Monte Carlo algorithms enable simulation from a posterior distribution, whilst only needing to access a sub-sample of data at each iteration.  We show how they can be implemented in settings where the parameters live on a restricted domain.  

\end{abstract}



\section{Introduction}
Markov chain Monte Carlo (MCMC) methods have been central to the wide-spread use of Bayesian methods. However their applicability to some modern applications has been limited due to their high computational cost, particularly in big-data, high-dimensional settings. This has led to interest in new MCMC methods, particularly non-reversible methods which can mix  better than standard reversible MCMC \cite{DiaconisHolmesNeal2000,TuritsynChertkovVucelja2011}, and variants of MCMC that require accessing only small subsets of the data at each iteration \cite{WellingTeh2011}.  

One of the main technical challenges associated with likelihood-based inference for big data is the fact that likelihood calculation is computationally expensive (typically $O(N)$ for data sets of size $N$). 
MCMC methods built from piecewise deterministic Markov processes (PDMPs) offer considerable promise for reducing this $O(N)$ burden, due to their ability to use sub-sampling techniques, whilst still being guaranteed to target the true posterior
distribution \cite{BierkensFearnheadRoberts2016,Bouchard-Cote2015,fearnhead2016piecewise,Galbraith2016,pakman2016}. Furthermore, factor graph decompositions of the target distribution can be leveraged to perform sparse updates of the variables \cite{Bouchard-Cote2015, Nishikawa2015,PetersDeWith2012}. 

PDMPs explore the state space according to constant velocity dynamics, but where the velocity changes at random event times. The rate of these event times, and the change in velocity at each event, are chosen so that the position of the resulting process has the posterior distribution as its invariant distribution. We will refer to this family of sampling methods as Piecewise Deterministic Monte Carlo methods (PDMC).


%

Existing PDMC algorithms can only be used to sample from posteriors where the parameters can take any value in $\reals^d$. In this paper (Section~\ref{sec:pdmc-restricted}) we show how to extend PDMC methodology to deal with constraints on the parameters. Such models are ubiquitous in machine learning and statistics. For example, many popular models used for binary, ordinal and polychotomous response data are multivariate real-valued latent variable models where the response is given by a deterministic function of the latent variables \cite{albert1993bayesian,fahrmeirmultivariate,train2009discrete}.
Under the posterior distribution, the domain of the latent variables is then constrained based on the values of the responses. Additional examples arise in regression where prior knowledge restricts the signs of marginal effects of explanatory variables such as in econometrics \cite{geweke1986exact}, image processing and spectral analysis \cite{bellavia2006interior}, \cite{guo2012comparison} and non-negative matrix factorization \cite{kim2007sparse}. 
A few methods for dealing with restricted domains are available but these either target an approximation of the correct distribution \cite{patterson2013} or are limited in scope \cite{pakman2014}.

\section{Piecewise Deterministic Monte Carlo on Restricted Domains}
\label{sec:pdmc-restricted}
Here we present the general PDMC algorithm in a restricted domain.
Specific implementations of PDMC algorithms can be derived as continuous-time limits of familiar discrete-time MCMC algorithms \cite{bierkens2015piecewise,PetersDeWith2012}, and these derivations convey much of the intuition behind why the algorithms have the correct stationary distribution. Our presentation of these methods is different, and more general. We first define a simple class of PDMPs and show how these can be simulated. We then 
give simple recipes for how to choose the dynamics of the PDMP so that it will have the correct stationary distribution.


Our objective is to compute expectations with respect to  a probability distribution $\pi$  
on $\mathcal{O}\subseteq \reals^d$ which is assumed to have a smooth density, also denoted $\pi(x)$, with respect to the Lebesgue measure on $\mathcal{O}$. 
With this objective in mind, we will construct a continuous-time Markov process $Z_t = (X_t, V_t)_{t \geq 0}$ taking values in the domain $E = \mathcal{O} \times \mathcal{V}$, where $\mathcal{O}$ and $\mathcal{V}$ are subsets of $\mathbb{R}^d$, such that $\mathcal{O}$ is open, pathwise connected and with Lipschitz boundary $\partial \mathcal O$. In particular, if $\mathcal O = \mathbb R^d$ then $\partial \mathcal O = \emptyset$. The dynamics of $Z_t$ are easy to describe if one views $X_t$ as position and $V_t$ as velocity. The position process $X_t$ moves deterministically, with constant velocity $V_t$ between a discrete set of \emph{switching times} which are simulated according to $N$ inhomogeneous Poisson processes, with respective intensity functions $\lambda_i(X_t, V_t)$, $i=1,\ldots, N$, depending on the current state of the system.  At each switching time the position stays the same, but the velocity is updated according to a specified transition kernel.  More specifically, suppose the next switching event occurs from the $i^{th}$ Poisson process, then the  velocity immediately after the switch is sampled randomly from the probability distribution $Q_i(x, v, \cdot)$ given the current position $x$ and velocity $v$.  The switching times are  random, and designed in conjunction with the kernels $(Q_i)_{i=1}^N$ so that the invariant distribution of the process coincides with the target distribution $\pi$. 

To ensure that $X_t$ remains confined within $\mathcal{O}$ the velocity of the process is updated whenever $X_t$ hits $\partial\mathcal{O}$ so that the process moves back into $\mathcal{O}$.  We shall refer to such updates as \emph{reflections} 
even though they need not be specular reflections.

The resulting stochastic process is a Piecewise Deterministic Markov Process (PDMP, \cite{davis1984piecewise}). For it to be useful as the basis of a Piecewise Deterministic Monte Carlo (PDMC) algorithm we need to (i) be able to easily simulate this process; and (ii) have simple recipes for choosing the intensities, $(\lambda_i)_{i=1}^N$, and  transition kernels, $(Q_i)_{i=1}^N$, such that the resulting process has $\pi(x)$ as its marginal stationary distribution. We will tackle each of these problems in turn.


\subsection{Simulation}
The key challenge in simulating our PDMP is simulating the event times. The intensity of events is a function of the state of the process. But as the dynamics between event times are deterministic, we can easily represent the intensity for the next event as a deterministic function of time. 
Suppose that the PDMP is driven by a single inhomogeneous Poisson process with intensity function
\[
\widetilde{\lambda}(u;X_t, V_t)=\lambda(X_t+ u  V_t, V_t), \quad u \geq 0.
\]
We can simulate the first event time directly if we have an explicit expression for the inverse function of the monotonically increasing function
\begin{equation}\label{eq:inversion}
u \mapsto \int_0^u \widetilde{\lambda}(s; X_t, V_t) \, ds.
\end{equation}
In this case the time until the next event is obtained by (i) simulating a realization, $y$ say, of an exponential random variable with rate $1$; and (ii) setting the time until the next event as the value $\tau$ that solves
$\int_0^{\tau} \widetilde \lambda (s;X_t,V_t) \, d s=y$.

Inverting~\eqref{eq:inversion} is often 
not practical. In such cases simulation can be carried out via  {\em thinning} \cite{lewis1979simulation}. This requires finding a tractable upper bound on the rate, $\overline \lambda(u)\geq \widetilde{\lambda}(u;X_t,V_t)$ for all $u>0$. 
Such an upper bound will typically take the form of a piecewise linear function or a step function. Note that the upper bound $\overline \lambda$ is only required to be valid along the trajectory $u \mapsto (X_t + u V_t, V_t)$ in $\mathcal O \times \mathcal V$. Therefore the upper bound may depend on the starting point $(X_t, V_t)$ of the line segment we are currently simulating.
We then propose potential events by simulating events from an inhomogenous Poisson process with rate $\overline \lambda(u)$, and accept an event at time $u$ with probability $\widetilde{\lambda}(u;X_t,V_t)/\overline \lambda(u)$. The time of the first accepted event will be the time until the next event in our PDMP.

To handle boundary reflections, at every given time $t$, we also keep track of the next reflection event in the absence of a switching event, i.e. we compute
\[
\tau_b = \inf\left\lbrace
u > 0\,:\, X_t + u V_t \not\in \mathcal{O}\right).
\]
If the boundary $\partial \mathcal{O}$ can be represented as a finite set of $M$ hyper-planes in $\mathbb{R}^d$, then the cost of computing $\tau_b$ is $O(Md)$.  When generating the switching event times and positions for $Z_t$ we determine whether a boundary reflection will occur before the next potential switching event.  If so, then we induce a switching event at time $t + \tau_b$ where $X_{t+\tau_b}\in\partial \mathcal{O}$ and sample a new velocity from the transition kernel $Q_b$, i.e. $V_{t+\tau_b} \sim Q_b(X_{t+\tau_b}, V_{t}, \cdot)$.

Although theoretically we may choose a new velocity pointing outwards and have an immediate second jump, we will for algorithmic purposes assume that the probability measure $Q_b(x,u, \cdot)$ for $(x,u) \in \partial \mathcal \mathcal{O} \times \mathcal V$ is concentrated on those directions $v$ for which $(v \cdot n(x)) \leq 0$, where $n(x)$ is the outward normal at $x \in \partial \mathcal{O}$.

For a PDMP driven by $N$ inhomogeneous Poisson processes with intensities $( \lambda_i )_{i=1}^N$ the previous steps lead to the following algorithm for simulating the next event of our PDMP. This algorithm can be iterated to simulate the PDMP for a chosen number of events or a pre-specified time-interval.
\begin{itemize}
\item[(0)] {\bf Initialize}: Set $t$ to the current time and  $(X_t,V_t)$ to the current position and velocity. 
\item[(1)] {\bf Determine bound}: For each $i \in 1,\ldots, N$, find a convenient function $\overline \lambda_i$ satisfying $\overline \lambda_i(u)\geq \widetilde{\lambda_i}(u;X_t, V_t)$ for all $u \geq 0$, depending on the initial point $(X_t,V_t)$ from which we are departing.
\item[(2)] {\bf Propose event}: For $i=1,\ldots, N$ simulate the first event times $\tau_i'$ of a Poisson process with rate function $\overline{\lambda}_i$. Compute the next boundary reflection time $\tau_b$.
\item[(3)] Let $i_{\min} = \argmin_{j=1,\ldots, N}\tau'_j$ and $\tau' = \tau'_{i_{min}}$.
\item[(4)] {\bf Accept/Reject event:} 
\begin{itemize}
\item[(4.1)] If $\tau_b < \tau'$ then set $\tau = \tau_b$; set $X_{t+\tau} = X_t + \tau V_t$; sample a new velocity $V_{t + \tau} \sim Q_b(X_{t+\tau}, V_t, \cdot)$. 
\item[(4.2)] Otherwise with probability 
\[
\frac{\widetilde{\lambda}_{i_{\min}}(\tau';X_t, V_t)}{\overline{\lambda}_{i_{\min}}(\tau')}
\]
accept the event at time $\tau = \tau'$. 
\begin{itemize}
\item[(4.2.1)] {\bf Upon acceptance}: set $X_{t+\tau}=X_t+\tau V_t$; sample a new velocity $V_{t+\tau} \sim Q_{i_{\min}}(X_{t+\tau},V_t,\cdot)$.
\item[(4.2.2)] {\bf Upon rejection}: set $X_{t+\tau} = X_t + \tau V_t$ and set $V_{t+\tau}=V_t$. 
\end{itemize}
\end{itemize}
\item[(5)] {\bf Update}: Record the time $t+\tau'$ and state $(X_{t+\tau'},V_{t+\tau'})$. 
\end{itemize}

\subsection{Output of PDMC algorithms}
The output of these algorithms will be a sequence of event times $t_1, t_2, t_3, \ldots, t_K$ and associated states $(X_1, V_1)$, $(X_2, V_2), \ldots, (X_K, V_K)$.  To obtain the value of the process at times $t \in [t_k, t_{k+1})$, we can linearly interpolate the continuous path of the process between event times, i.e. $X_t = X_{t_k} + V_k(t-t_k)$.  Time integrals $\int_0^t f(X_s)\,ds$ of a function $f$ of the process $X_t$ can often be computed analytically from  the output of the above algorithm. If not they can be
approximated by numerically integrating the one dimensional integral along the piecewise linear trajectory of the PDMP.  
Alternatively we can sample the PDMP at a set of evenly spaced time points along the trajectory and use this collection as an approximate sample from our target distribution.

Under the assumption that the resulting PDMP is ergodic (for sufficient conditions see e.g. \cite{BierkensFearnheadRoberts2016,Bouchard-Cote2015}) and that the marginal density on $\mathcal O$ of the stationary distribution of $(X_t, V_t)$ is equal to $\pi$, we have the following version of the law of large numbers 
for the PDMP $(X_t,V_t)_{t \geq 0}$: For all $f \in L^2(\pi)$ we have that, with probability one,
\[ \int_{\reals^d} f(x) \pi(x) \, d x = \lim_{T \rightarrow \infty} \frac 1 T \int_0^T f(X_s) \, d s.\]
It is this formula which allows us to use PDMPs for Monte Carlo purposes.


\subsection{Choosing the intensity and transition kernels}\label{sec:choosing-intensity-transition}

Assume, as most existing PDMC methods do \cite{BierkensFearnheadRoberts2016,Bouchard-Cote2015,PetersDeWith2012}, that the target density, $\pi(x) : \mathcal O \rightarrow (0,\infty)$ is differentiable. Under this condition we can provide criteria on the switching intensities $(\lambda_i)$ and transition kernels $Q_i$ and $Q_b$ which must hold for a given probability distribution to be a stationary distribution of $Z_t$. We shall consider stationary distributions for which $x$ and $v$ are independent, i.e. distributions of the form 
$\pi(x) d x \otimes \rho(dv)$ on $E$. Furthermore we assume that 
$\pi(x) \propto \exp(-U(x))$
where $U$ is continuously differentiable.

We impose the condition that
\begin{align}
\label{eq:switching_condition}
\int_{v \in \mathcal V} \sum_{i=1}^{N}\lambda_i(x,v) Q_i(x,v,du) \, \rho(dv)  = \int_{v \in \mathcal V}\sum_{i=1}^{N}\lambda_i(x,v) Q_i(x,u,dv) \, \rho(du), \quad x \in \mathcal O.
\end{align}
A sufficient condition for~\eqref{eq:switching_condition} is that each $Q_i$ is reversible with respect to $\rho$, i.e. for every $i=1,\ldots, N$ and $x \in \mathcal O$, we have that $Q_i(x, v, du)\rho(dv) = Q_i(x, u, dv)\rho(du)$.

Moreover, we shall require the following condition which relates the probability flow with the switching intensities $\lambda_i$:
\begin{equation}
\label{eq:intensity_condition}
\begin{aligned}
\sum_{i=1}^N\int_{\mathcal{V}} \lambda_i(x,v)Q_i(x,u, dv) -\sum_{i=1}^N\lambda_i(x,u)= -u\cdot\nabla U(x), \quad (x,u) \in \mathcal O \times \mathcal V.
\end{aligned}
\end{equation}

Finally, the boundary transition kernel should satisfy
\begin{equation}
\label{eq:q_boundary_reversible}
Q_b(x, u, dv)\rho(du) = Q_b(x,v,du)\rho(dv), \quad x \in \partial \mathcal O,
\end{equation}
and 
\begin{equation}
\label{eq:q_boundary_assumption}
\int_{\mathcal{V}} (n(x)\cdot u)\, Q_b(x,v,du) = - v\cdot n(x), \quad (x,v) \in \partial \mathcal O \times \mathcal V,
\end{equation}
where for $x\in\partial \mathcal{O}$, we denote by $n(x)$ the outward unit normal of $\partial \mathcal{O}$.


\begin{prop}
\label{prop:boundary-case}
Consider the process $Z_t$ on $\mathcal O \times \mathcal{V}$ where $\mathcal{O}$ is an open, pathwise connected subset of $\mathbb{R}^d$ with Lipschitz boundary $\partial \mathcal O$. Suppose that conditions~\eqref{eq:switching_condition},\eqref{eq:intensity_condition}, \eqref{eq:q_boundary_reversible} and~\eqref{eq:q_boundary_assumption} are satisfied.
Then $\pi(x) \, dx \otimes \rho(dv)$ is an invariant distribution for the process $Z_t$.
\end{prop}

The proof of this result relies on verifying that $\mathbb E_{\pi\otimes \rho}[\mathcal Lf(X,V)]=0$ where $\mathcal L$ denotes the generator of our PDMP and is deferred to the supplementary material, Section 1.

In practice we only have to satisfy~\eqref{eq:q_boundary_reversible} and~\eqref{eq:q_boundary_assumption} on the exit region $\Gamma \subset \mathcal O \times \mathcal V$. For example if $\mathcal O = (a,b) \subset \reals^1$ and $\mathcal V = \{-1,+1\}$, then  $\Gamma = \{ b, +1 \} \cup \{a, -1\}$. The specification of $Q_b$ on $(\partial \mathcal O \times \mathcal V) \setminus \Gamma$ is irrelevant as these points are never reached by $Z_t$. On this irrelevant set, we may choose $Q_b$ as desired to satisfy~\eqref{eq:q_boundary_assumption}. 




\subsection{Example: The Bouncy Particle Sampler}
Current PDMC algorithms differ in terms of how the $Q_i$ and $\lambda_i$ are chosen such that the above equation holds for some simple distribution for the velocity. Here we discuss how the Bouncy Particle Sampler (BPS), introduced in \cite{PetersDeWith2012} and explored in \cite{Bouchard-Cote2015}, is an example of the framework introduced here. In the supplementary material, Section 1.1, the Zig-Zag sampler is described as a second example. In the following example $\delta_x$ denotes the Dirac-measure centered in $x$.

The Bouncy Particle Sampler is obtained setting $N = 1$ and  $\rho=\mathcal{N}\left(0,I \right)$ on $\mathbb{R}^d$ or $\rho=\mathcal{U}(S^{d-1})$, i.e. the uniform distribution on the unit sphere. The single switching rate is chosen to be $\lambda_{\text{BPS}}(x,v) = \max(v \cdot \nabla U(x), 0)$, with corresponding switching kernel $Q$ which reflects $v$ with respect to the orthogonal complement of $\nabla U$ with probability 1:
\[ Q(x, v, dv') = \delta_{(I - 2 P_{\nabla U})v}(dv'),\]
where $P_y: z \mapsto \frac{z \cdot y}{\|y\|^2} y$ denotes orthogonal projection along the one dimensional subspace spanned by $y$.

As noted in \cite{Bouchard-Cote2015} this algorithm suffers from reducibility issues. These can be overcome by refreshing the velocity by drawing a new velocity independently from $\rho(dv)$. In the simplest case the refreshment times come from an independent Poisson process with constant rate $\lambda_{\text{ref}}$. This also fits in the framework above by choosing $\widetilde{\lambda}=\lambda_{\text{BPS}}+\lambda_{\text{ref}}$ and 
\begin{align*} Q(x,u,dv) =\frac{\lambda_{\text{BPS}}}{\lambda_{\text{BPS}} +\lambda_{\text{ref}}}\delta_{(I - 2 P_{\nabla U})u}(dv)  +\frac{\lambda_{\text{ref}}}{\lambda_{\text{BPS}}+\lambda_{\text{ref}}} \rho(dv).\end{align*}

As boundary transition kernel it is natural to choose  
\[
Q_b(s,v,du)=\delta_{(I-2P_{n(s)})v}(du),
\]
for $s \in \partial \mathcal{O}$, so that the process $X_t$ reflects specularly at the boundary (i.e. angle of incidence equals angle of reflection of process with respect to the boundary normal). It is straightforward to check that condition \eqref{eq:switching_condition} holds at the boundary and that \eqref{eq:q_boundary_assumption} is satisfied. 

As a generalization of the BPS, one can consider a \emph{preconditioned} version, which is 
 obtained by introducing a constant positive definite symmetric matrix $M$ to rescale the velocity process.
The choice of $M$ plays a very similar role to the mass matrix in HMC, and careful tuning can give rise to dramatic increases in performance \cite{Galbraith2016,pakman2016}.  





\section{Subsampling}
\label{sec:subsampling}
When using PDMC to sample from a posterior, we can use sub-samples of data at each iteration of the algorithm, as described in \cite{BierkensFearnheadRoberts2016,Bouchard-Cote2015}, which  reduces the computational complexity of the algorithm from $O(N)$ to $O(1)$, where $N$ is the size of the data, without affecting the theoretical validity of the algorithm. In the following we will assume that we can write the posterior as
$
\pi(x)\propto \prod_{i=1}^N f(\data_i;x),
$
for some function $f$. For example this would be the likelihood for a single IID data point times the $1/N$th power of the prior.

The idea of using sub-sampling, within say the Bouncy Particle Sampler (BPS), is that at each iteration of our PDMC algorithm we can replace $\nabla U(x)$ by an unbiased estimator in step (3). We need to use the same estimate both when calculating the actual event rate in the accept/reject step and, if we accept, when simulating the new velocity. 
The only further alteration we need to the algorithm is to choose an upper bound $\overline{\lambda}$ that holds for all realizations of $\widehat{\nabla U}$.  A more comprehensive explanation of this argument can be found in \cite{BierkensFearnheadRoberts2016,fearnhead2016piecewise} in the context of the Zig-Zag sampler, and in \cite{Bouchard-Cote2015,Galbraith2016} for the bouncy particle sampler.

We first present a way for estimating $\nabla U$ unbiasedly using control variates 
\cite{Bardenet:2015,BierkensFearnheadRoberts2016}.  For any $x, \hat{x}\in \mathcal{O}$ we note that $
\nabla U(x)= \nabla U(\hat{x})+\left[\nabla  U(x)- \nabla  U(\hat{x}) \right]$.
We can then introduce the estimator $\widehat{\nabla U}(x)$ of $\nabla U(x)$ by 
\begin{equation}\label{eqn:subsampling}
\textstyle
\widehat{\nabla {U}}(x)=\nabla U(\hat{x})+N \left[ 
\nabla \log f(\data_{I};x)-
\nabla \log f(\data_{I};\hat{x})
\right],
\end{equation}\label{gradientestimate}
where $I$ is drawn uniformly from $\{1,\ldots,N\}$.

It is straightforward to show that the resulting BPS algorithm uses an event rate that is
$\mathbb{E}\left[\max\left(0,(\widehat{\nabla {U}}(x)\cdot v)\right)\right]$, and that this rate and the resulting transition probability $Q$ at events satisfies Proposition~\ref{prop:boundary-case}. Hence this algorithm still targets $\pi(x)$, but only requires access to one data point at each accept-reject decision. 

Note that this gain in computational efficiency does not come for free, as it follows from Jensen's inequality that the overall rate of events will be higher. This makes mixing of the PDMC process slower. It is also immediate that the bound, $\overline{\lambda}$, we will have to use will be higher. However \cite{BierkensFearnheadRoberts2016} show that if our estimator of $\widehat{\nabla U}(x)$ has sufficiently small variance, then we can still gain substantially in terms of efficiency. In particular they give an example where the CPU cost effective sample size does not grow with $N$ -- by comparison all standard MCMC algorithms would have a cost that is at least linear in $N$.

To obtain such a low-variance estimator requires a good choice of $\hat{x}$, so that with high probability $x$ will be closer to $\hat{x}$. This involves a preprocessing step to find a value $\hat{x}$ close to the posterior mode, a preprocessing step to then calculate $\nabla U(\hat{x})$ is also needed.

We now illustrate how to find an upper bound on the event rate. Following \cite{BierkensFearnheadRoberts2016}, if we assume $L$ is a uniform (in space and $i$) upper bound on the largest eigenvalue of the Hessian of $U^i$, and if $\| v \| = 1$:
\begin{align}
& \max\left(0, \left(\nabla U(\hat{x})+N(\nabla U^{i}(X_t)-\nabla U^{i}(\hat{x}))\right)\cdot v\right) \nonumber 
\\& 
\leq \max\left(0,\nabla U(\hat{x})\cdot v\right)+N\left\| \nabla U^{i}(x)-\nabla U^{i}(\hat{x})
\nabla U^{i}(x)-\nabla^{i}U(X_t)\right\| \nonumber \\
 & \leq \max\left(0,\nabla U(\hat{x})\cdot v\right)+NL\left\| x-\hat{x}\right\| +N Lt\label{eq:ControlVariate}
\end{align}
Thus the upper bound on the intensity is of the form $\bar{\lambda}(\tau)=a+b\cdot\tau$ with $a,b\ge 0$. In this case the first arrival time can be simulated as follows
\begin{equation}
\label{eq:simulating-linear-bound} \tau'=-a/b+\sqrt{\left(\frac{a}{b}\right)^2+2\cdot\frac{R}{b} } \text{ with } R \sim \text{Exp}(1).\end{equation}
An alternative and complementary approach to improve the efficiency of this subsampling procedure is to use an estimator of the gradient \eqref{gradientestimate} where $I$ is drawn according to a distribution dependent on the observations \cite{Bouchard-Cote2015,Galbraith2016}.






\section{Software and Numerical Experiments}
A open-source Julia package \texttt{PDMP.jl} has been developed to provide efficient implementations of various recently developed piecewise deterministic Monte Carlo methods for sampling in (possibly restricted) continuous spaces.  A variety of algorithms are implemented including the Zig-Zag sampler and the Bouncy Particle Sampler with full and local refreshment along with control variate based sub-sampling  for these methods.  The package has been specifically designed with extensibility in mind, permitting rapid implementation of new PDMP based methods. The library along with code and documentation is available at \url{github.com/alan-turing-institute/PDSampler.jl}.

We use Bayesian binary logistic regression as a testbed for our newly proposed methodology and perform a simulation study. The 
data $\data_{i}\in\{-1,1\}$ is modelled by 
\begin{equation}
p(\data_{i}\vert \cova_{i},\param)=\logit(y_{i}\param^{T}\cova_{i})\label{eq.logistic}
\end{equation}
where $\cova \in \mathbb{R}^{p\times n}$ are fixed covariates and $\logit(z)=\frac{1}{1+\exp(-z)}\in[0,1]$. We will assume that we wish to fit this model under some monotonicity constraints -- so that the probability of $y=1$ is known to either increase or decrease with certain covariates. This is modeled through the constraint  $\param_i>0$ and $\param_i<0$ respectively. 
An example where such restrictions occur naturally is in logistic regression for questionnaires, see  \cite{Tutz2014}.
In following we consider the case $\param_j \ge 0$ for $j=1,\dots,p$ along with the additional linear constraint  $\sum_{j} x_j \leq K$ where $K=10$.

For simplicity we use a flat prior over the space of parameters values consistent with our constraints. By Bayes' rule the posterior $\pi$ satisfies 
\[
\pi(\param)\propto \prod_{i=1}^{N}\logit(y_{i}\param^{T}\cova_{i}) \mbox{ for } \param\in \mathcal{O},
\]
where $\mathcal{O}$ is the space of parameter values consistent with our constraints.
We implement the BPS with subsampling. As explained in the introduction, subsampling is a key benefit of using  piecewise deterministic sampling methods; see Section~\ref{sec:subsampling}.
We use reflection at the boundary i.e. $Q_b(s,v,du)=\delta_{(I-2P_{n(s)})v}(du)$ for $s\in\partial\mathcal{O}$. 
We can bound the switching intensity by a linear function of time, even when we use the subsampling estimator for the switching rate. See the supplementary material, Section 2, for details on the application of subsampling in this example. 
We use $n=10,000$ and $p=20$ and generate artificial data based on $\cova$ and $\param^\star$  whose components are a realization i.i.d. of uniformly distributed random variables satisfying the imposed constraints.  

We compare the performance of BPS to standard MALA and HMC schemes, in terms of effective sample size (ESS) per epoch of data evaluation. 
For each scheme we obtain the distribution of ESS based on $10$ independent realisations of each chain.
In Figure \ref{Fig:Logistic}(a) we plot for each scheme, the distribution of ESS per epoch with respect to the function $f_1(x) = \frac{1}{p}(x_1 + \ldots + x_p)$.   Similarly, In Figure \ref{Fig:Logistic}(b) we plot the ESS per epoch for each chain with respect to the function $f_2(x) = \log \pi(x)$.   The performance of MALA and HMC appears commensurate and the BPS demonstrates a clear advantage over both in terms of ESS per epoch.

The HMC and MALA schemes were tuned by minimising the ESS with respect to the step-size, calculated from exploratory runs. For HMC we use $5$ leap-frog steps.  We find that we must tune both HMC and MALA to have a small step size due to proposals being rejected at the boundary. The ESS is estimated based on asymptotic variance using the batch means method; see \cite{BierkensDuncan2016,BierkensFearnheadRoberts2016} for details. 


For specific types of constraints more efficient implementations of HMC and MALA are possible, either by introducing an appropriate transformation of the restricted state space, or by reflecting the posterior distribution along the constraint boundaries.   Moreover, we note that there exists a version of HMC which can sample from truncated Gaussian distributions \cite{pakman2014}.  However, to our knowledge there is no efficient HMC or MALA scheme able to handle generally restricted domains.

The Bouncy Particle Sampler for this model was implemented using \texttt{PDSampler.jl} while the corresponding HMC and MALA samplers implemented with \texttt{Klara.jl}.  The code for this numerical experiment along with results are carefully presented in \url{github.com/tlienart/ConstrainedPDMP/}.



\begin{figure}[t!]
\centering
    \begin{subfigure}[t]{0.45\textwidth}
    \centering  
        \includegraphics[width=\textwidth]{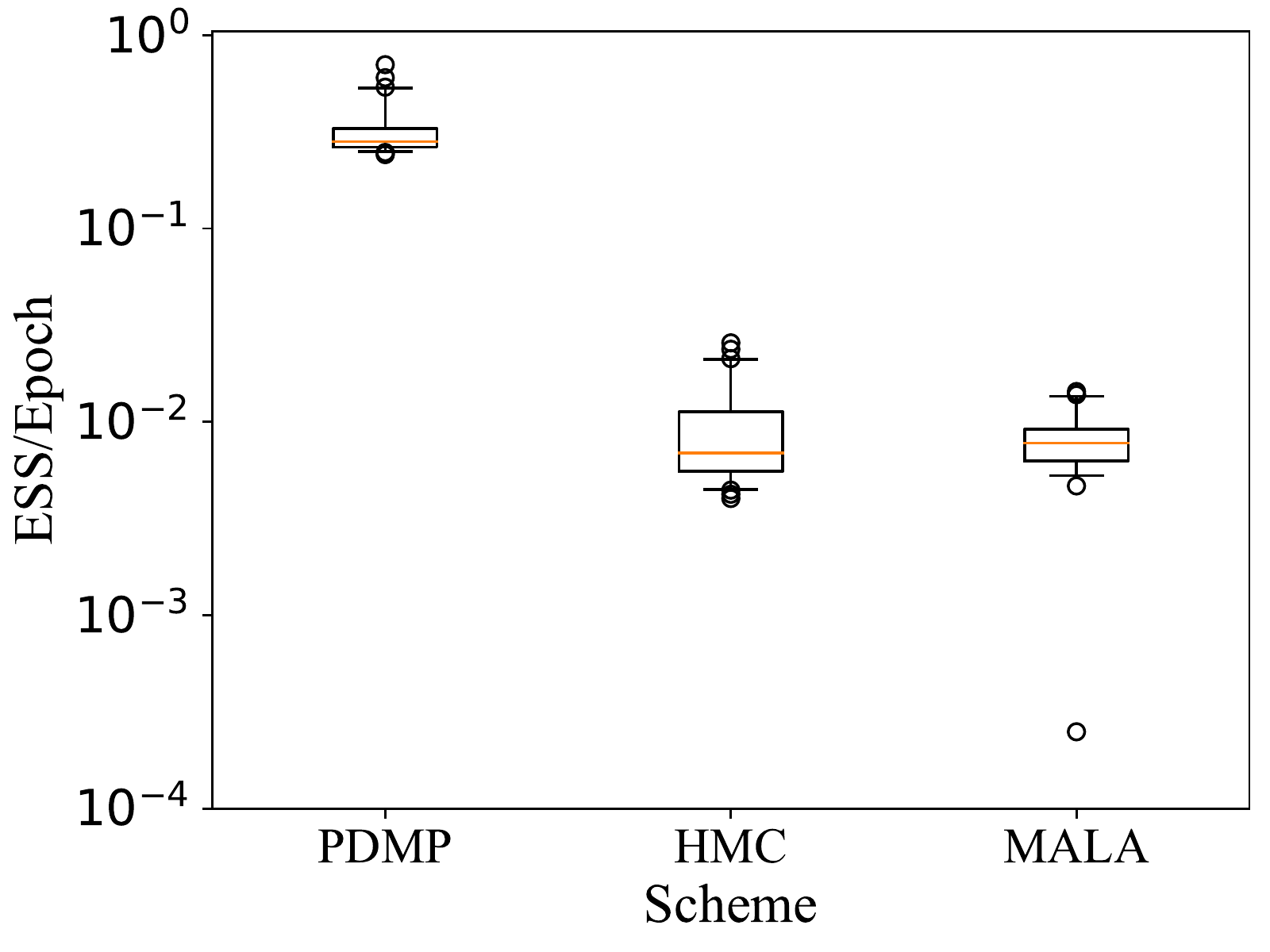}
        \caption{ESS  per epoch with respect to $f_1(x)$}
    \end{subfigure}
    ~
   \begin{subfigure}[t]{0.45\textwidth}
    \centering
        \includegraphics[width=\textwidth]{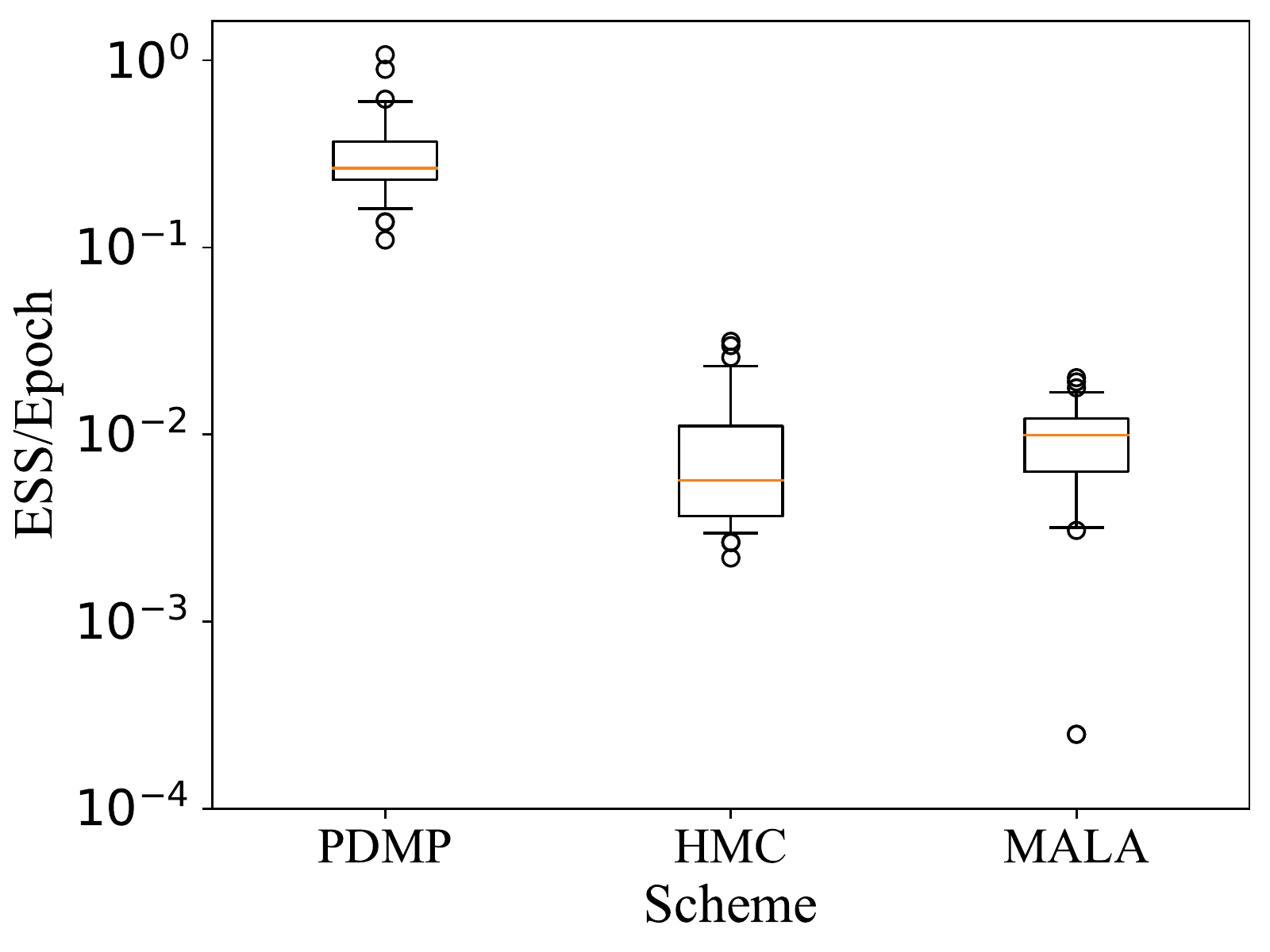}
        \caption{ESS per epoch with respect to $f_2(x)$}
    \end{subfigure}
\caption{Average ESS per epochs of data evaluation for  MALA, HMC and PDMP (BPS) applied to logistic regression with $p=20$ and $n=10000$ and parameter $\param$ constrained to be nonnegative and satisfy $\sum_j x_j \leq 10$. The graphic is based on 10 independent runs for each HMC, MALA and BPS for each choice of number of epochs. }
\label{Fig:Logistic}

\end{figure}



\section{Discussion}

This work provides a framework for describing a general class of PDMC methods which are ergodic with respect to a given target probability distribution. Open questions remain on how the choice of intensity function, velocity transition kernel as well as other parameters of the system influence the overall performance of the scheme.  The problem of understanding the true computational cost of such PDMC schemes is more subtle than for classical discrete time MCMC schemes: often one needs to find a balance between fast mixing of the continuous time Markov process and having a switching rate that is relatively cheap to simulate. For example, when using subsampling the mixing of the Markov process is  slower than without subsampling, but the computational cost per simulated switch is significantly smaller.
Further investigation is required to understand this delicate balance.

\section*{Acknowledgements}
All authors thank the Alan Turing Institute and Lloyds registry foundation for support. S.J.V. gratefully acknowledges funding through EPSRC EP/N000188/1.  J.B., P.F. and G.R. gratefully acknowledge EPSRC ilike grant EP/K014463/1.  A.B.D. acknowledges grant EP/L020564/1 and A.D. acknowledges grant EP/K000276/1. We kindly acknowledge comments from the editor and anonymous reviewers which have significantly improved the exposition in this paper. 

 \bibliographystyle{plain}
 \bibliography{references.bib}

\newpage
\onecolumn

\pagebreak
\begin{center}
\textbf{\large Supplement to \emph{Piecewise Deterministic Markov Processes for Scalable Monte Carlo on Restricted Domains}}
\end{center}
\setcounter{equation}{0}
\setcounter{figure}{0}
\setcounter{table}{0}
\setcounter{page}{1}
\setcounter{section}{0}
\makeatletter
\renewcommand{\theequation}{S\arabic{equation}}
\renewcommand{\thefigure}{S\arabic{figure}}


In the supplementary material we provide a theoretical background for the framework for restricted domain PDMC (Section \ref{sec:supp-stationary}, which includes the Zig-Zag sampler as further example). 
The detailed application of subsampling to the logistic regression example may be found in Section~\ref{sup:Logistic}.

%

\section{Stationary distribution for PDMPs on restricted domains}
\label{sec:supp-stationary}
From  \cite[Section 5]{davis1984piecewise} the process $Z_t$ will have infinitesimal generator given by the closure of the operator 
\begin{equation}
\label{eq:generator}
\mathcal{L}f(x,v) = v\cdot\nabla_x f(x,v)+\sum_{i=1}^{N}\lambda_i(x,v)\int_{\mathcal{V}}(f(x,u) - f(x,v))Q_i(x,v,du),\quad (x,v)\in E,
\end{equation}
where $\mathcal{D}(\mathcal L)$ is the set of functions which are continuously differentiable with respect to $x$ on $\mathcal O$, which is decaying to infinity as $\lVert x \rVert\rightarrow \infty$ and such that   
\begin{equation}
\label{eq:boundary_condition}
f(x,v) = \int_{\mathcal{V}} f(x,u)Q_b(x,v,du),
\end{equation}
for all $(x, v)\in \partial\mathcal{O}\times \mathcal{V}$.
Based on this identification of the infinitesimal generator we can now provide a formal proof that the conditions of Proposition 1 of the paper are sufficient to ensure invariance of $\pi \otimes \rho$.  

\begin{proof}[Sketch Proof of Proposition 1]
Without loss of generality we take $\pi(x)=\exp(-U(x))$, i.e. the proportionality factor in $\pi(x) \propto \exp(-U(x))$ is assumed to be 1. We shall only provide a formal proof of this result, by demonstrating that
\[
\int_{\mathcal O \times \mathcal V} \mathcal{L}f(x,v)\pi(x)\,dx \,\rho(dv) = 0, \quad \mbox{ for all }f \in \mathcal{D}(\mathcal{L}),
\]
so that $\mathcal{L}$ is infinitesimally invariant.  A rigorous proof would require establishing that  $\mathcal{D}(\mathcal{L})$ as defined above is a core for the extended generator.  This is a technical result which we defer for future work. 

For $f \in \mathcal{D}(\mathcal L)$,
\begin{align*}
 & \int_{\mathcal O}\int_{\mathcal{V}}\int_{\mathcal{V}}\sum_{i=1}^{N}\lambda_i(x,v)\left[f(x,u)-f(x,v)\right]\, Q_i(x,v,du) \pi(x)\, dx \, \rho(dv)\\
 & =\int_{\mathcal O}\int_{\mathcal{V}}\sum_{i=1}^{N} f(x,u)\left[\int_{\mathcal{V}}\lambda_i(x,v)Q_i(x,u,dv)-\lambda_i(x,u)\right]\, \rho(du) \pi(x)\, dx\\
 & =-\int_{\mathbb{R}^{d}}\int_{\mathcal{V}}f(x,u)u\cdot\nabla U(x)\, \rho(du)e^{-U(x)}\, dx\\
 & =\int_{\mathcal O}\int_{\mathcal{V}}f(x,u)u\cdot\nabla e^{-U(x)}\, \rho(du)\,dx\\
 & =-\int_{\mathcal O}\int_{\mathcal{V}}u\cdot\nabla_x f(x,u)e^{-U(x)}\,\rho(du)\,dx +\int_{\partial \mathcal O}\int_{\mathcal{V}}f(\sigma,u)(u\cdot n(\sigma))e^{-U(\sigma)}\rho(du)\, d \sigma,
\end{align*}
where the boundary term arises from integration by parts with respect to $x$.  Considering the boundary integral, by applying (4) (in the paper) which is assumed to hold on $\partial\mathcal{O}$ and  (\ref{eq:boundary_condition}) (above) we obtain
\begin{align*}
&\int_{\partial \mathcal O}\int_{\mathcal{V}}f(\sigma,u)(u\cdot n(\sigma)) e^{-U(\sigma)}\rho(du)\,d\sigma\\ & =\int_{\partial \mathcal O}\int_{\mathcal{V}}\int_{\mathcal{V}}f(\sigma,v)Q_b(\sigma,u,dv)(u\cdot n(\sigma))e^{-U(\sigma)}\rho(du)\,d\sigma \\
 & =\int_{\partial \mathcal O}\int_{\mathcal{V}}\int_{\mathcal{V}}f(\sigma,v)Q_b(\sigma,v,du)(u\cdot n(\sigma))e^{-U(\sigma)}\rho(dv)\,d\sigma\\
 & =-\int_{\partial\mathcal O}\int_{\mathcal{V}}f(\sigma,v)(v\cdot n(\sigma))e^{-U(\sigma)}\rho(dv)\,d\sigma,
\end{align*}
so that the boundary term evaluates to zero. 

It follows that 
\begin{align*}
\int_{\mathcal O}\int_{\mathcal{V}}\mathcal{L}f(x,v)\pi(x)\, dx\rho(dv)=\int_{\mathcal O}\int_{\mathcal{V}}\left(u\cdot\nabla_xf(x,u) - u\cdot\nabla_x f(x,u)\right)\pi(x)\, dx\rho(dv) = 0,
\end{align*}
so that $\pi(x)\, d x \otimes\rho(dv)$ is infinitesimally invariant with respect to $Z_{t}.$
\end{proof}






Another possible behaviour at the boundary is to generate the new reflected direction independently of the angle of incidence.  This will also preserve the invariant distribution provided that $\rho$ is isotropic.

\begin{prop}
\label{prop:boundary_independent}
Consider the process $Z_t$ as in the previous proposition,  such that  conditions (2) and (3) (of the paper) hold and the distribution $\rho$ has mean zero.  Then  $\pi(x) \, dx \otimes \rho(dv)$ will be an invariant distribution for the process $Z_t$ if $Q_b(x, v, du)$ is independent of $v$  for all  $x\in\partial\mathcal{O}$.
\end{prop}

\begin{proof}[Sketch Proof of Proposition \ref{prop:boundary_independent}]
Let $f \in \mathcal{D}(\mathcal{L})$, so that $f$ satisfies (\ref{eq:boundary_condition}).  By the assumptions on $Q_b$ in Proposition 1 (of the paper), this implies that $f(x,v) = f(x)$ for all $x \in \partial \mathcal{O}$.   Following the proof of Proposition 1 above, the boundary integral term becomes
\begin{align*}
\int_{\partial \mathcal O}\int_{\mathcal{V}}f(s,u)(u\cdot n(s)) e^{-U(s)}\rho(du)\,ds =\int_{\partial \mathcal O}f(s)e^{-U(s)}n(s)\,ds\cdot\int_{\mathcal{V}}u\rho(du),
\end{align*}
which is zero if $\rho$ has mean zero, as required.
\end{proof}

\section{The Zig-Zag sampler}
\label{sec:zigzag}

The Zig-Zag sampler \cite{BierkensFearnheadRoberts2016} can be recovered by choosing $N = d$ and picking as velocity space $\mathcal V = \{-1,+1\}^d$ equipped with discrete uniform distribution $\rho$, defining switching rates $\lambda_i(x,v) = \max(v_i \partial_{x_i} U(x),0)$. The corresponding switching kernels over new directions are given by
\[ Q_i(x,v, dv') = \delta_{F_i v} (dv'),\]
where $F_i : \mathcal V \rightarrow \mathcal V$ denotes the operation of flipping the $i$-th component, i.e. $(F_i v)(i) = -v(i)$, and $(F_i v)(j) = v(j)$ for $j \neq i$.

\section{Derivation of dominating intensity for logistic regression example}
\label{sup:Logistic}

A valid choice of $L$ can be derived as follows: 
Notice that $\left(\log\logit(z)\right)^{\prime}=\logit(-z)$ and
$\logit^{'}(z)=\logit(z)\left(1-\logit(z)\right)$ so that we obtain
\begin{align*}
\frac{\partial}{\partial x}\log\logit(\data_{i}\vert x) & =\logit\left(-\data_{i}x^{\top}\cova_{i}\right)\data_{i}\cova_{i}\\
\frac{\partial}{\partial^{2}x}\log\logit(\data_{i}\vert x) & =-\logit\left(-\data_{i}x^{\top}\cova_{i}\right)\left(1-\logit\left(-\data_{i}x^{\top}\cova_{i}\right)\right)\cova_{i}\cova_{i}^{\tau}
\end{align*}
Using $p(1-p)\leq\frac{1}{4}$ for $p\in[0,1]$
\[
\sup_{\left\Vert w\right\Vert \leq1}\left|w^{t}\frac{\partial}{\partial^{2}x}\log\logit(\data_{i}\vert x)w^{t}\right|\leq\frac{1}{4}\left\Vert \cova_{i}\right\Vert ^{2}
\]
So Equation (7) (of the paper) holds with $$L:=\frac{1}{4}\max_{i=1,\dots,n} \left \Vert \cova_i \right\Vert $$ as defined
above. 

This is a linear upper bound on the intensity which can be used to sample according to (8) (of the paper) and then used for thinning as introduced in Section 1 of the paper.

\end{document}